\documentclass[12pt,letterpaper]{article}
\usepackage{amsmath,amsfonts,amsthm,amssymb,stmaryrd,graphicx,relsize,bm}
\theoremstyle{plain}

\numberwithin{equation}{section}
\newtheorem{thm}{Theorem}[section]
\newtheorem{lem}[thm]{Lemma}
\newtheorem{cor}[thm]{Corollary}

\allowdisplaybreaks  % introduced 7.15.15

\newcounter{cond}

\newcommand{\complex}{{\mathbb C}}

\newcommand{\real}{{\mathbb R}}

\newcommand{\ascript}{{\mathcal A}}
\newcommand{\bscript}{{\mathcal B}}
\newcommand{\cscript}{{\mathcal C}}
\newcommand{\dscript}{{\mathcal D}}
\newcommand{\escript}{{\mathcal E}}
\newcommand{\fscript}{{\mathcal F}}
\newcommand{\hscript}{{\mathcal H}}

\newcommand{\sscript}{{\mathcal S}}

\newcommand{\shat}{\widehat{S}}
\newcommand{\ahat}{\widehat{a}}
\newcommand{\bhat}{\widehat{b}}
\newcommand{\chat}{\widehat{c}}
\newcommand{\dhat}{\widehat{d}}
\newcommand{\alphabar}{\overline{\alpha}}
\newcommand{\ascripthat}{\widehat{\ascript}}
\newcommand{\omegahat}{\widehat{\omega}}
\newcommand{\phihat}{\widehat{\phi}}

\newcommand{\ctimes}{\mathrel{\mathlarger\cdot}}
\newcommand{\rmtr}{\mathrm{tr}}

\newcommand{\ab}[1]{\left|#1\right|}
\newcommand{\brac}[1]{\left\{#1\right\}}
\newcommand{\paren}[1]{\left(#1\right)}
\newcommand{\sqbrac}[1]{\left[#1\right]}
\newcommand{\elbows}[1]{{\left\langle#1\right\rangle}}

\begin{document}

\title{CONVEX AND SEQUENTIAL\\
EFFECT ALGEBRAS}
\author{Stan Gudder\\ Department of Mathematics\\
University of Denver\\ Denver, Colorado 80208\\
sgudder@du.edu}
\date{}
\maketitle

\begin{abstract}
We present a mathematical framework for quantum mechanics in which the basic entities and operations have physical significance. In this framework the primitive concepts are states and effects and the resulting mathematical structure is a convex effect algebra. We characterize the convex effect algebras that are classical and those that are quantum mechanical. The quantum mechanical ones are those that can be represented on a complex Hilbert space. We next introduce the sequential product of effects to form a convex sequential effect algebra. This product makes it possible to study conditional probabilities and expectations.
\end{abstract}

%Keywords and Phrases:
% effect algebras; sequential products; quantum mechanics; hidden variables
% Classifications: quantum mechanics

\section{Introduction}  % Section 1
One of the most important problems in the foundations of physics is to justify the axioms of quantum mechanics on physical grounds. A simplified version of the main axioms of quantum mechanics is the following.
\begin{list} {(A\arabic{cond})}{\usecounter{cond}
\setlength{\rightmargin}{\leftmargin}}
%(A1)
\item The pure states of a quantum system are represented by unit vectors in a complex Hilbert space $K$ and the observables are represented by self-adjoint operators on $K$.
%(A2)
\item If the system is in the state $\phi$ then the expectation (or average value) of an observable $A$ is $\elbows{\phi ,A\phi}$
%(A3)
\item The dynamics of the system is described by a one-parameter unitary group $U_t$, $t\in\real$. If the initial state is $\phi _0$ then the state at time $t$ is $U_t\phi _0$.
\end{list}

Several immediate questions come to mind. Where does the complex Hilbert space come from? In particular, what do complex numbers have to do with a physical system? What is the physical meaning of the complex inner product $\elbows{\phi ,\psi}$? Two observables are said to be compatible if their corresponding operators $A,B$ commute. This is reasonable because if $A$ and $B$ commute they are both functions of another self-adjoint operator so they can be measured simultaneously. But if $A$ and $B$ do not commute, there is no physical meaning for the operator sum $A+B$ and the operator product $AB$. There are many other problems and questions like these. We conclude that these axioms are based upon unphysical structures whose basic mathematical operations have no physical meaning.

In this article we present a mathematical framework for quantum mechanics in which the basic entities and operations have physical significance. In this framework the primitive concepts are states and effects. The states represent initial preparations that describe the condition of the system while the effects represent yes-no measurements that probe the system. The effects may be unsharp or as they are sometimes called, fuzzy \cite{bug94,gud981,gud982}. A state applied to an effect produces the probability that the effect gives a yes value. Effects can also be thought of as true-false or 0-1 measurements. The resulting mathematical structure is a convex-effect algebra
$\escript$ \cite{gud99, gp98}. The two mathematical operations in $\escript$ are an orthogonal sum $a\oplus b$ and a scalar product 
$\lambda a,\lambda \in\sqbrac{0,1}\subseteq\real$ both of which having physical interpretations. The sum $a\oplus b$ corresponds to a parallel measurement of $a$ and $b$ while $\lambda a$ corresponds to an attenuation of $a$ by the factor $\lambda$
\cite{gud99, gp98}. Section~2 presents these basic definitions in detail.

One advantage of employing physically motivated mathematical operations is that they lead up to physically useful theorems and results. Our main theorems in Section~3 characterize the convex effect algebras that are classical and those that are quantum mechanical. The quantum mechanical convex effect algebras are those that can be represented on a complex Hilbert space and this answers the question: Where does the Hilbert space come from? The key to the representation theorem is a concept we call contextuality as explained in Section~3.

In Section~4 we introduce the sequential product $a\circ b$ of effects $a$ and $b$. This product corresponds to first measuring $a$ and then measuring $b$ in sequence. This product makes it possible to study conditional probabilities and expectations which are treated in Section~4. The resulting structure is called a convex sequential effect algebra \cite{gg02, gg05,gl08, gn01}.

\section{Convex Effect Algebras} % Section 2
Most statistical theories for physical systems contain two basic primitive concepts, namely effects and states. The effects correspond to simple yes-no measurements or experiments and the states correspond to preparation procedures that specify the initial conditions of the system being measured. Usually, each effect $a$ and state $s$ experimentally determine a probability $F(a,s)$ that the effect $a$ occurs (has answer yes) when the system has been prepared in the state $s$. For a given physical system, denote its set of possible effects by $\escript$ and its set of possible states by $\sscript$. In a reasonable statistical theory, the probability function satisfies three axioms that are given in the following definition  \cite{gud99}.

An \textit{effect-state space} is a triple $(\escript ,\sscript ,F)$ where $\escript$ and $\sscript$ are nonempty sets and
$F:\escript\times\sscript\to\sqbrac{0,1}\subseteq\real$ satisfies:
\begin{list} {(ES\arabic{cond})}{\usecounter{cond}
\setlength{\rightmargin}{\leftmargin}}
%(ES1)
\item There exist elements $0,1\in\escript$ such that $F(0,s)=0$, $F(1,s)=1$ for every $s\in\sscript$.
%(ES2)
\item If $F(a,s)\le F(b,s)$ for every $s\in\sscript$, then there exists a unique $c\in\escript$ such that $F(a,s)+F(c,s)=F(b,s)$ for every
$s\in\sscript$.
%(ES3)
\item If $a\in\escript$ and $\lambda\in\sqbrac{0,1}\subseteq\real$, then there exists an element $\lambda a\in\escript$ such that $F(\lambda a,s)=\lambda F(a,s)$ or every $s\in\sscript$.
\end{list}

The elements $0,1$ in (ES1) correspond to the null effect that never occurs and the unit effect that always occurs, respectively. Condition (ES2) postulates that if $a$ has a smaller probability of occurring than $b$ in every state, then there exists a unique effect $c$ which when combined with $a$ gives the probability that $b$ occurs in every state. The element $\lambda a$ of condition (ES3) is interpreted as the effect $a$ attenuated by the factor $\lambda$. It is shown in \cite{gud99} that if $F(a,s)+F(b,s)\le 1$ for every
$s\in\sscript$, then there exists a unique $c\in\escript$ such that $F(c,s)=F(a,s)+F(b,s)$ for every $s\in\sscript$. We then write
$a\perp b$ and define $a\oplus b=c$.

We now consider a previously studied mathematical framework that exposes the basic axioms of an effect-state space. An
\textit{effect algebra} \cite{dp94, fb94, kop92, kc94} is an algebraic system $(\escript ,0,1,\oplus )$ where $0$ and $1$ are distinct elements of $\escript$ and $\oplus$ is a partial binary operation on $\escript$ that satisfies the following conditions (we write $a\perp b$ when $a\oplus b$ is defined).
\begin{list} {(E\arabic{cond})}{\usecounter{cond}
\setlength{\rightmargin}{\leftmargin}}
%(E1)
\item If $a\perp b$, then $b\perp a$ and $b\oplus a=a\oplus b$.
%(E2)
\item If $a\perp b$ and $(a\oplus b)\perp c$, then $b\perp c$, $a\perp (b\oplus c)$ and $a\oplus (b\oplus c)=(a\oplus b)\oplus c$.
%(E3)
\item For every $a\in\escript$ there exists a unique $a'\in\escript$ such that $a\perp a'$ and $a\oplus a'=1$.
%(E4)
\item If $a\perp 1$, then $a=0$.
\end{list}

If $a\perp b$, we call $a\oplus b$ the \textit{orthogonal sum} of $a$ and $b$. We define $a\le b$ if there exists $c\in\escript$ such that
$a\oplus c=b$. It can be shown that $(\escript ,0,1,\le )$ is a bounded poset and $a\perp b$ if and only if $a\le b'$ \cite{fb94}. It is also shown in \cite{fb94} that $a''=a$ and that $a\le b$ implies $b'\le a'$ for every $a,b\in\escript$.

An effect algebra $\escript$ is \textit{convex} \cite{gud99,gp98} if for every $a\in\escript$ and $\lambda\in\sqbrac{0,1}\subseteq\real$ there exists an element $\lambda a\in\escript$ such that the following conditions hold.
\begin{list} {(C\arabic{cond})}{\usecounter{cond}
\setlength{\rightmargin}{\leftmargin}}
%(C1)
\item If $\alpha ,\beta\in\sqbrac{0,1}$ and $a\in\escript$, then $\alpha (\beta a)=(\alpha\beta )a$.
%(C2)
\item If $\alpha ,\beta\in\sqbrac{0,1}$ with $\alpha +\beta\le 1$ and $a\in\escript$, then $\alpha a\perp\beta a$ and
$(\alpha +\beta )a=\alpha a\oplus\beta a$.
%(C3)
\item If $a,b\in\escript$ with $a\perp b$ and $\lambda\in\sqbrac{0,1}$, then $\lambda a\perp\lambda b$ and
$\lambda (a\oplus b)=\lambda a\oplus\lambda b$.
%(C4)
\item If $a\in\escript$, then $1a=a$.
\end{list}

It is shown in \cite{gud99} that a convex effect algebra is ``convex'' in the sense that $\lambda a\oplus (1-\lambda )b$ is defined for every $\lambda\in\sqbrac{0,1}$ and $a,b\in\escript$ and hence is an element of $\escript$. If $\escript$ and $\fscript$ are effect algebras, a map $\phi\colon\escript\to\fscript$ is \textit{additive} if $a\perp b$ implies $\phi (a)\perp\phi (b)$ and
$\phi (a\oplus b)=\phi (a)\oplus\phi (b)$. An additive map $\phi$ that satisfies $\phi (1)=1$ is called a \textit{morphism}. A morphism
$\phi\colon\escript\to\fscript$ for which $\phi (a)\perp\phi (b)$ implies that $a\perp b$ is called a \textit{monomorphism}. A surjective monomorphism is an \textit{isomorphism}. If $\escript$ and $\fscript$ are convex effect algebras, a morphism
$\phi\colon\escript\to\fscript$ is an \textit{affine morphism} if $\phi (\lambda a)=\lambda\phi (a)$ for every $\lambda\in\sqbrac{0,1}$,
$a\in\escript$. If there exists an affine isomorphism $\phi\colon\escript\to\fscript$ we say that $\escript$ and $\fscript$ are
\textit{affinely isomorphic}.

The simplest example of a convex effect algebra is the unit interval $\sqbrac{0,1}\subseteq\real$ with the usual addition
(when $a+b\le 1$) and scalar multiplication. A \textit{state} on an effect algebra $\escript$ is a morphism
$\omega\colon\escript\to\sqbrac{0,1}$. We interpret $\omega (a)$ as the probability that the effect $a$ occurs when the system is prepared in the state $\omega$. We denote the set of states on $\escript$ by $\Omega (\escript )$. We say that
$S\subseteq\Omega (\escript )$ is \textit{separating} if $\omega (a)=\omega (b)$ for every $\omega\in S$ implies that $a=b$. We say that $S\subseteq\Omega (\escript )$ is \textit{order determining} if $\omega (a)\le\omega (b)$ for all $\omega\in S$ implies that $a\le b$. It is shown in \cite{gud99} that every state on a convex effect algebra is affine. The next result, which is proved in \cite{gud99} shows that an effect-state space is equivalent to a convex effect algebra with an order determining set of states. It is surprising that the physically motivated framework of an effect-state space with three simple axioms is equivalent to a seemingly more complicated structure of a convex effect algebra with an order determining set of states which has nine axioms.

\begin{thm}    % Theorem 2.1
\label{thm21}
If $(\escript ,S,F)$ is an effect-state space and $\shat =\brac{F(\ctimes ,s)\colon s\in S}$, then $(\escript ,0,1,\oplus )$ is a convex effect algebra with an order determining set of states $\shat$. Conversely, if $(\escript ,0,1,\oplus )$ is a convex effect algebra and $S$ is an order determining set of states on $\escript$, then $(\escript ,S,F)$ is an effect-state space where $F\colon\escript\times S\to\sqbrac{0,1}$ is defined by $F(a,s)=s(a)$.
\end{thm}

We now consider a general type of convex effect algebra called a linear effect algebra. Let $V$ be a real linear space with zero $\theta$. A subset $K$ of $V$ is a \textit{positive cone} if $\real ^*K\subseteq K$, $K+K\subseteq K$ and $K\cap (-K)=\brac{\theta}$. For
$x,y\in V$ we define $x\le y$ if $y-x\in K$. Then $\le$ is a partial order on $V$ and we call $(V,K)$ an \textit{ordered linear space} with positive cone $K$. We say that $K$ is \textit{generating} if $V=K-K$. Let $u\in K$ with $u\ne\theta$ and form the interval
\begin{equation*}
\sqbrac{\theta ,u}=\brac{x\in K\colon x\le u}
\end{equation*}
For $x,y\in\sqbrac{\theta ,u}$ we write $x\perp y$ if $x+y\le u$ and in this case we define $x\oplus y=x+y$. It is clear that
$\paren{\sqbrac{\theta ,u},\theta ,u,\oplus }$ is an effect algebra with $x'=u-x$ for every $x\in\sqbrac{\theta ,u}$. It is easy to check that
$\sqbrac{\theta ,u}$ is a convex subset of $V$ and that $\lambda x\in\sqbrac{\theta ,u}$ for every $\lambda\in\sqbrac{0,1}$,
$x\in\sqbrac{\theta ,u}$. It follows that $\sqbrac{\theta ,u}$ is a convex effect algebra which we call a \textit{linear effect algebra}. We say that $\sqbrac{\theta ,u}$ \textit{generates} $K$ if $K=\real ^+\sqbrac{\theta ,u}$ and say that $\sqbrac{\theta ,u}$ is \textit{generating} if $\sqbrac{\theta ,0}$ generates $K$ and $K$ generates $V$. The following representation theorem, which is proved in \cite{gp98}, shows that convex effect algebras and linear effect algebras are equivalent structures.

\begin{thm}    % Theorem 2.2
\label{thm22}
If $(\escript ,0,1,\oplus )$ is a convex effect algebra, then $\escript$ is affinely isomorphic to a linear effect algebra $\sqbrac{\theta ,u}$ that generates an ordered linear space $(V,K)$.
\end{thm}

A linear functional $f\colon V\to\real$ on an ordered linear space $(V,K)$ is \textit{positive} if $f(x)\ge 0$ for all $x\in K$. We denote the set of positive linear functionals on $V$ by $V^p$. If $\sqbrac{\theta ,u}$ generates $(V,K)$ and $f\in V^p$ satisfies $f(u)=1$ we say that $f$ is \textit{unital}. We denote the set of unital elements of $V^p$ as $V_u^p$. It is clear that if $f\in V_u^p$, the the restriction of $f$ to
$\sqbrac{\theta ,u}$ is a state. The next result, which is proved in \cite{gud99} gives a converse.

\begin{thm}    % Theorem 2.3
\label{thm23}
Let $\sqbrac{\theta ,u}$ be a generating interval for $(V,K)$.\newline
{\rm (i)}\enspace If $\omega\in\Omega\paren{\sqbrac{\theta ,u}}$, then $\omega$ has a unique extension $\omegahat\in V_u^p$.
\newline
{\rm (ii)}\enspace The map ${}^\wedge\colon\Omega\paren{\sqbrac{\theta ,u}}\to V_u^p$ is a bijection that satisfies
\begin{equation*}
\paren{\lambda\omega _1+(1-\lambda )\omega _2}^\wedge =\lambda\omegahat _1+(1-\lambda )\omegahat _2
\end{equation*}

for all $\lambda\in\sqbrac{0,1}$, $\omega _1,\omega _2\in\Omega\paren{\sqbrac{\theta ,u}}$.\newline
{\rm (iii)}\enspace A subset $S\subseteq\Omega\paren{\sqbrac{\theta ,u}}$ is order determining if and only if $\shat\subseteq V_u^p$ is order determining.
\end{thm}

Of course, $\shat$ order determining means that $\omegahat (x)\le\omegahat (y)$ for all $\omega\in S$ implies that $x\le y$. We close this section with two important examples of convex effect algebras. The first example comes from the quantum theory formalism \cite{kra83,lud83}. Let $H$ be a complex Hilbert space and let $\escript (H)$ be the set of operators on $H$ that satisfy $0\le A\le I$ where we are using the usual ordering of bounded operators. For $A,B\in\escript (H)$ we write $A\perp B$ if $A+B\in\escript (H)$ and in this case we define $A\oplus B=A+B$. For $\lambda\in\sqbrac{0,1}$ and $A\in\escript (H)$, $\lambda A\in\escript (H)$ is the usual scalar multiplication for operators. It is easy to check that $\escript (H)$ is a convex effect algebra which we call a
\textit{Hilbertian effect algebra}. If $\phi\in H$ is a unit vector, define the state $\phihat$ by $\phihat (A)=\elbows{\phi ,A\phi}$ for all
$A\in\escript (H)$. It follows by definition that this set of states is order determining.

Our second example comes from fuzzy probability theory \cite{bug94,gud981}. Let $(\Omega ,\ascript )$ be a measurable space in which singleton sets are measurable and let $\escript (\Omega ,\ascript )$ be the set of measurable functions on $\Omega$ with values in $\sqbrac{0,1}\subseteq\real$. If we define $\oplus$ and $\lambda f$ analogously as in the previous example, we see that
$\escript (\Omega ,\ascript )$ is a convex effect algebra. The elements of $\escript (\Omega ,\ascript )$ are called \textit{fuzzy events} and we call $\escript (\Omega ,\ascript )$ a \textit{classical effect algebra}. If $\mu$ is a probability measure on $(\Omega ,\ascript )$ then the map $f\mapsto\int fd\mu$ gives a state on $\escript (\Omega ,\ascript )$. This set of states is order determining. In particular, the set of Dirac measures $\delta _\omega$, $\omega\in\Omega$ is order determining.

\section{Classical and Hilbertian Effect Algebras} % Section 3
This section characterizes the classical and Hilbertian effect algebras. Roughly speaking, these correspond to classical and quantum mechanics, respectively. For simplicity, we only treat the finite-dimensional case. Our theory generalizes to infinite dimensions but then we have to treat $\sigma$-effect algebras \cite{gud982}. This would introduce measure theoretic and convergence details that detract from the main ideas. Besides there are important physical systems such as quantum information and computation that fall within the finite dimensional domain.

Let $\escript$ be a convex effect algebra. By Theorem~\ref{thm22} we can assume that $\escript$ is a linear effect algebra $\sqbrac{\theta ,u}$ that generates an ordered linear space $(V,K)$. For $x,y\in V$ we sometimes retain the notation $x\oplus y$ if
$x,y\in\escript =\sqbrac{\theta ,u}$ with $x\perp y$ and otherwise we use $x+y$ for the sum. An effect $a\in\escript$ is \textit{sharp} \cite{gud982} if the greatest lower bound $a\wedge a'=\theta$. Sharp effects are thought of as effects that are precise or unfuzzy. The sharp effects in $\escript (\Omega ,\ascript )$ are the measurable characteristic functions or equivalently the sets in $\ascript$. The sharp effects in $\escript (H)$ are the projection operators on $H$. We denote the sharp effects in $\escript$ by $S(\escript )$. An
$a\in S(\escript )$ is \textit{one -dimensional} if $a\ne\theta$ and if $b\in\escript$ with $b\le a$ implies that $b=\lambda a$ for some
$\lambda\in\sqbrac{0,1}$. It is shown in \cite{gp98} that if $a\in S(\escript )$ with $a\ne\theta$ then there exists a state
$\ahat\in\Omega (\escript )$ such that $\ahat (a)=1$. We denote the set of one-dimensional sharp elements by $S_1(\escript )$.

A \textit{context} is a finite set $\brac{a_1,\ldots ,a_n}\subseteq S_1(\escript )$ such that
\begin{equation}                % equation (3.1)
\label{eq31}
a_1\oplus a_2\oplus\cdots\oplus a_n=u
\end{equation}
It follows from \eqref{eq31} that $\ahat _i(a_j)=\delta _{ij}$. We interpret a context as a finest sharp measurement. That is, one of the effects $a_i$ must occur and there is no finer sharp measurement. We say that $\escript$ is \textit{finite-dimensional} if there exits a context on $\escript$. For the remainder of this section, we shall assume that $\escript$ is finite-dimensional. We say that $\escript$ is \textit{spectral} if for every $b\in\escript$ there exists a context $\brac{a_1,\ldots ,a_n}$ such that
$b=\lambda _1a_1\oplus\cdots\oplus\lambda _na_n$, $\lambda _i\in\sqbrac{0,1}$, $i=1,\ldots ,n$. We now characterize a classical effect algebra $\escript (\Omega ,\ascript )$. We say that $\escript (\Omega ,\ascript )$ is \textit{finite} if
$\Omega =\brac{\omega _1,\ldots ,\omega _n}$ is finite.

\begin{thm}    % Theorem 3.1
\label{thm31}
Let $\escript$ be a finite dimensional convex effect algebra. Then $\escript$ is affinely isomorphic to a finite classical effect algebra if and only if $\escript$ possesses exactly one context and $\escript$ is spectral.
\end{thm}
\begin{proof}
For sufficiency, we can assume that $\escript =\escript (\Omega ,\ascript )$ where $\Omega =\brac{\omega _1,\ldots ,\omega _n}$ is finite. A function $f\in\escript$ is sharp if and only if $f$ has the values $0$ or $1$; that is, $f$ is a characteristic function. Indeed, characteristic functions are clearly sharp. Conversely, suppose $f\in\escript$ is sharp and $f(\omega _0)\ne 0,1$ for some $\omega _0\in\Omega$. Let $\lambda\in (0,1)$ satisfy $\lambda<f(\omega _0)$, $\lambda <1-f(\omega _0)$. Define $g\in\escript$ by $g(\omega _0)=\lambda$, $g(\omega )=0$ if $\omega\ne\omega _0$. Then $g<f$ and $g<1-f=f'$. Since $g\ne 0$, $f\wedge (1-f)\ne 0$. This gives a contradiction so $f$ is a characteristic function. The functions in $S_1(\escript )$ are the characteristic functions of singleton sets
$\chi _{\brac{\omega}}$, $\omega\in\Omega$. Since
\begin{equation*}
\chi _{\brac{\omega _1}}\oplus\cdots\oplus\chi _{\brac{\omega _n}}=1
\end{equation*}
we see that $\brac{\chi _{\brac{\omega}}\colon\omega\in\Omega}$ is the only context in $\escript$. Also every $f\in\escript$ has the form
$f=\sum\lambda _i\chi _{\brac{\omega _i}}$, $\lambda\in\sqbrac{0,1}$ so $\escript$ is spectral. Conversely, suppose $\escript$ has a single context $\brac{a_1,\cdots ,a_n}$ and $\escript$ is spectral. Let $(\Omega ,\ascript )$ be a finite measurable space with
$\Omega =\brac{\omega _1,\cdots ,\omega _n}$ For $b=\lambda_1a_1\oplus\cdots\oplus\lambda _na_n\in\escript$, define
$J(b)\in\escript (\Omega ,\ascript )$ by $J(b)(\omega _i)=\lambda _i$. Then $J\colon\escript\to\escript (\Omega ,\ascript )$ is bijective, $J(1)=1$, $J(\lambda b)=\lambda J(b)$. If $b\perp c$ with $c=\mu _1a_1\oplus\cdots\oplus\mu _na_n\in\escript$ we have
\begin{align*}
b\oplus c&=(\lambda _1+\mu _1)a_1\oplus\cdots\oplus (\lambda _n+\mu _n)a_n\\
\intertext{and}
J(b\oplus c)(\omega _i)&=\lambda _i+\mu _i=J(b)(\omega _i)+J(c)(\omega _i)
\end{align*}
$i=1,\ldots ,n$, so $J(b\oplus c)=J(b)+J(c)$. Finally, if $J(b)\perp J(c)$ we have that $J(b)(\omega _i)+J(b)(\omega _i)\le 1$,
$i=1,\ldots ,n$. Hence, $\lambda _i+\mu _i\le 1$, $i=1,\ldots ,n$ so $b\perp c$. We conclude that $J$ is an affine isomorphism.
\end{proof}

If $\ascript =\brac{a_i\colon i=1,\ldots ,n}$ is a context on the convex effect algebra $\escript$, we form the set of states
$\ascripthat =\brac{\ahat _i\colon i=1,\ldots ,n}$. It follows from Theorem~\ref{thm23} that $\ascripthat$ can be thought of as a set of positive, unital, linear functionals on $(V,K)$. We now construct the complex linear space
\begin{equation*}
\hscript (\ascript )=\brac{\sum _{i=1}^n\alpha _i\ahat _i\colon\alpha _i\in\complex}
\end{equation*}

For $x,y\in\hscript (\ascript )$ with $x=\sum\alpha _i\ahat _i$, $y=\sum\beta _i\ahat _i$ we define the inner product $\elbows{x,y}=\sum\alphabar _i\beta _i$. Thus, $\hscript (\ascript )$ is a complex Hilbert space that we call the \textit{state space for the context}
$\ascript$. Of course, $\hscript (\ascript )$ is $n$-dimensional with orthonormal basis $\ascripthat =\brac{\ahat _i\colon i=1,\ldots ,n}$.

Now $\ascripthat$ naturally generates a real linear space of linear functions on $(V,K)$ so why did we choose $\hscript (\ascript )$ to be a complex rather than a real space? One reason is that we need to describe a dynamics for states in $\hscript (\ascript )$. Since a dynamics must preserve norms and orthogonality, it is represented by a continuous group of unitary operators
$U_i\colon\hscript (\ascript )\to\hscript (\ascript )$, $t\in\real$, for context $\ascript$. It is assumed that $U_{t_1+t_2}=U_{t_1}U_{t_2}$ so the operators $U_t$ commute. Thus, they are simultaneously diagonalizable and hence have common eigenvectors $\phi _i\in\hscript (\ascript )$ so that
\begin{equation*}
U_t\phi _i=\alpha _i(t)\phi _i
\end{equation*}
$i=1,\ldots ,n$. If $\hscript (\ascript )$ is a real Hilbert space, then $\alpha _i(t)\in\real$ and since $U_t$ is unitary $\alpha _i(t)=\pm 1$. But then $U_t$ cannot be continuous unless $U_t=I$ for all $t$. In the complex case, $\alpha _i(t)=e^{i\theta _i(t)},\theta _i(t)\in\real$, which is continuous if $\theta _i(t)$ is continuous, $i=1,2,\ldots ,n$. In fact, we have $\alpha _i(t)=e^{i\theta _it}$. In this case, denoting the one-dimensional projection onto $\ahat$, by $P(\ahat _i)$, we have the Hamiltonian $L=\sum\theta _iP(\ahat _i)$ so that $U_t=e^{iLt}$. There are also other groups such as rotations that require unitary representations on a complex Hilbert space of states.

Notice that the one-dimensional effects are atoms among the sharp effects. Indeed, if $a$ is one-dimensional and $b\in\escript$ with $0<b<a$, then $b=\lambda a$, $\lambda\in (0,1)$. If $\mu <\lambda$, $\mu <1-\lambda$, then $\mu a<\lambda a$ and since
\begin{align*}
(u+\lambda )a&<(\mu +\lambda )u<u\\
\intertext{we have that}
\mu a&<u-\lambda a=(\lambda a)'
\end{align*}
Hence,
\begin{equation*}
b\wedge b'=(\lambda a)\wedge (\lambda a)'\ne 0\quad\hbox{whether or not it exists.}
\end{equation*}
Since $b\notin S(\escript )$, there are no nonzero sharp elements strictly below $a$ so $a$ is an atom in $S(\escript )$.

If $\ascript =\brac{a_i\colon i=1,\ldots ,n}$ is a context and $b\in\escript$ define the linear operator $b _{\ascript}$ on $\hscript (\ascript )$ by
\begin{equation*}
b_{\ascript}\sum\alpha _i\ahat _i=\sum\alpha _i\ahat _i(b)\ahat _i
\end{equation*}

\begin{lem}    % Lemma 3.2
\label{lem32}
The map $J\colon\escript\to\escript\paren{\hscript (\ascript )}$ given by $J(b)=b_{\ascript}$ is an affine morphism.
\end{lem}
\begin{proof}
Since $J(b)\ahat _i=\ahat _i(b)\ahat _i$, we see that $J(b)$ is a positive linear operator with eigenvalues $0\le\ahat _i(b)\le 1$ and corresponding eigenvectors $\ahat _i$. Thus $J(b)\in\escript\paren{\hscript (\ascript )}$. Also, $J(\theta )=0$, $J(u)=I$ and we have
\begin{align*}
J(b\oplus c)\sum\alpha _i\ahat _i&=\sum\alpha _i\ahat _i(b\oplus c)\ahat _i=\sum\alpha _i\sqbrac{\ahat _i(b)+\ahat _i(c)}\ahat _i\\
&=\paren{J(b)+J(c)}\sum\alpha _i\ahat _i
\end{align*}
Hence, $J(b\oplus c)=J(b)+J(c)$ so $J$ is a morphism. Since $J(\lambda b)=\lambda J(b)$, $\lambda\in\sqbrac{0,1}$, $J$ is affine.
\end{proof}

The affine morphism $J(b)=b_{\ascript}$ of Lemma~\ref{lem32} gives a representation of $\escript$ into the Hilbertian effect algebra
$\escript\paren{\hscript (\ascript )}$. However, $J$ need not be injective or surjective and $J$ need not preserve sharpness. Moreover, all the $J(b)$, $b\in\escript$, commute so they do not convey quantum interference. One can say that $J$ gives a distorted partial view of $\escript$. The reason for this is that we are only employing a single context $\ascript$. Unlike a classical effect algebra with only one context, a quantum effect algebra has many contexts. Each gives a partial view and in order to obtain a total view, they must all be considered.

In order to consider several contexts together, we introduce a method to compare them. A collection of contexts
$\Gamma =\brac{\ascript ,\bscript ,\cscript,\ldots}$ is \textit{comparable} if for every $\ascript ,\bscript\in\Gamma$ there exists a unitary transformation $U_{\ascript\bscript}\colon\hscript (\ascript )\to\hscript (\bscript )$ such that $U_{\ascript\ascript}=I$,
$U_{\ascript\bscript} =U_{\bscript\ascript}^*$ and if $a\in\ascript$, $c\in\cscript$ then
\begin{equation}                % equation (3.2)
\label{eq32}
\ab{\elbows{U_{\ascript\bscript}\ahat ,U_{\cscript\bscript}\chat\,}}^2=\ahat (c)
\end{equation}
We call $\ahat (c)$ in \eqref{eq32} the \textit{transition probability} from $a$ to $c$. In particular, we can compare the elements of
$\ascript$ and $\bscript$ together by
\begin{equation*}
\ab{\elbows{U_{\ascript\bscript}\ahat ,\bhat}}^2=\ab{\elbows{U_{\ascript\bscript}\ahat,U_{\bscript\bscript}\bhat}}^2=\ahat (b)
\end{equation*}

Notice that a unit vector $\phi$ in $\hscript (\ascript )$ can be considered as a vector in the Hilbert space $\hscript (\ascript )$ or as a state on $\escript$, where the state corresponding to $\phi$ is $\phihat$ given by
\begin{equation*}
\phihat (b)=\elbows{\phihat ,b_{\ascript}\phihat\,}
\end{equation*}
This is consistent with $\ahat (b)=\elbows{\ahat ,b_{\ascript}\ahat}$ for all $a\in\ascript$. A collection of contexts
$\Gamma =\brac{\ascript ,\bscript ,\cscript ,\ldots}$ is \textit{complete} if they are comparable and if for any $\bscript\in\Gamma$ and any unit vector $\phi\in\hscript (\bscript )$ there exists an $\ascript\in\Gamma$ and an $a\in\ascript$ such that
$U_{\ascript\bscript}(\ahat )=\phi$.

As an example, in the classical case there is only one context $\ascript$. Then $\ascript$ is comparable with $U_{\ascript\ascript}=I$. But $\ascript$ is not complete unless $\ascript =\brac{1}$ and $\hscript (\ascript )=\complex$; that is, $\hscript (\ascript )$ is
one-dimensional. Indeed, suppose $\ascript$ is complete and $\ascript =\brac{a_1,\ldots ,a_n}$. If $\phi\in\hscript (\ascript )$ with
$\phi =\frac{1}{\sqrt{2}}(\ahat _1+\ahat _2)$ then there exists $a_j\in\ascript$ such that
\begin{equation*}
\ahat _j=U_{\ascript\ascript}(\ahat _j)=\phi
\end{equation*}
But this is impossible unless $\ascript =\brac{a_j}$ and so $a_j=1$. We conclude that $\escript$ is affinely isomorphic to
$\sqbrac{0,1}\subseteq\real$ and $\hscript (\ascript )=\complex$.

\begin{thm}    % Theorem 3.3
\label{thm33}
Let $\escript$ be a finite dimensional convex effect algebra. Then $\escript$ is affinely isomorphic to a Hilbertian effect algebra if and only if its set of contexts is complete and $\escript$ is spectral.
\end{thm}
\begin{proof}
To prove necessity we can assume that $\escript =\escript (H)$ for some Hilbert space $H$. The elements of $S_1(\escript )$ become one-dimensional projections and it follows from the spectral theorem that $\escript (H)$ is spectral. Since $\escript$ is finite dimensional, every context has the form $\ascript =\brac{a_1,\ldots ,a_n}$ where $a_i\in S_1(\escript )$. Thus, $a_i$ is a projection onto the subspace of $H$ spanned by a unit vector $\phi _i$ where $\brac{\phi _1,\ldots ,\phi _n}$ is an orthonormal basis for $H$. We can then identify
$\ascripthat$ with this basis. It is now straightforward to show that the set of contexts of $\escript$ is complete. Conversely, suppose that the set of contexts for $\escript$ is complete and $\escript$ is spectral. Letting $\bscript$ be a fixed context we shall show that $\escript$ is affinely isomorphic to $\escript\paren{\hscript (\bscript )}$. If $b\in\escript$, since $\escript$ is spectral, we have that
$b=\sum\lambda _ia_i$, $\lambda _i\in\sqbrac{0,1}$ for some context $\ascript =\brac{a_i\colon i=1,\ldots ,n}$. Now
$\brac{\ahat _i\colon i=1,\ldots ,n}$ forms an orthonormal basis for $\hscript (\ascript )$ and since $U_{\ascript\bscript}$ is unitary, 
$\brac{U_{\ascript\bscript} (\ahat _i)\colon i=1,\ldots ,n}$ is an orthonormal basis for $\hscript (\bscript )$. Let $P(a_i)$ be the
one-dimensional projection onto the subspace of $\hscript (\bscript )$ spanned by $U_{\ascript\bscript}(\ahat _i)$. Define
$J\colon\escript\to\escript\paren{\hscript (\bscript )}$ by $J(b)=\sum\lambda _iP(a_i)$. To show that $J$ is additive, suppose
$c\in\escript$ with $c\perp b$ and $c=\sum\mu _ic_i$ for some context $\cscript =\brac{c_i}$. Since $\escript$ is spectral,
$b\oplus c=\sum\gamma _id_i$ for some context $\dscript =\brac{d_i}$. We then have that
\begin{equation}                % equation (3.3)
\label{eq33}
\sum\gamma _id_i=\sum\lambda _ia_i+\sum\mu _ic_i
\end{equation}
If $\phi$ is a unit vector in $\hscript (\bscript )$, there exists a $d\in S_1(\escript )$ and a context $\fscript$ with $d\in\fscript$ and
$U_{\fscript\bscript}\dhat =\phi$. Applying $\dhat$ to \eqref{eq33} gives
\begin{equation}                % equation (3.4)
\label{eq34}
\sum\gamma _i\dhat (d_i)=\sum\lambda _i\dhat (a_i)+\sum\mu _i\dhat (c_i)
\end{equation}
Since the contexts are comparable, applying \eqref{eq34} and \eqref{eq32} gives
\begin{align*}
\sum\gamma _i\ab{\elbows{U_{\fscript\bscript}\dhat ,U_{\dscript\bscript}\dhat _i}}^2
&\!=\!\sum\lambda _i\ab{\elbows{U_{\fscript\bscript}\dhat ,U_{\ascript\bscript}\ahat _i}}^2
\!+\!\sum\mu _i\ab{\elbows{U_{\fscript\bscript}\dhat ,U_{\cscript\bscript}\chat _i}}^2\\
\intertext{Hence,}
\sum\gamma _i\elbows{U_{\fscript\bscript}\dhat ,P(d_i)U_{\fscript\bscript}\dhat\,}
&\!=\!\sum\lambda _i\elbows{U_{\fscript\bscript}\dhat ,P(a_i)U_{\fscript\bscript}\dhat\,}\\
&\quad+\sum\mu _i\elbows{U_{\fscript\bscript}\dhat ,P(c_i)U_{\fscript\bscript}\dhat\,}\\
\intertext{which gives}
\elbows{\phi ,J(b\oplus c)\phi}&=\elbows{\phi ,J(b)\phi}+\elbows{\phi ,J(c)\phi}
\end{align*}
Since the pure states of $\hscript (\bscript )$ are separating we conclude that $J(b\oplus c)=J(b)+J(c)$ so $J$ is additive. To show that $J$ is affined, let $b=\sum\lambda _ia_i$. Then $\lambda b=\sum\lambda\lambda _ia_i$, $\lambda\in\sqbrac{0,1}$ and we obtain
\begin{equation*}
J(\lambda b)=\sum\lambda\lambda _iP(a_i)=\lambda J(b)
\end{equation*}
It is clear that $J$ has a unique linear extension to $V$. We leave it to the reader to show that $J$ is injective. To show that $J$ is surjective, let $P_\phi$ be a one--dimensional projection onto the subspace of $\hscript (\bscript )$ spanned by the unit vector $\phi$. By completeness, there is an $a\in S_1(\escript )$ with $J(a)=P_\phi$. If $A\in\escript\paren{\hscript (\bscript )}$ has spectral decomposition
$A=\sum\lambda _iP_{\phi _i}$ we have $a_i\in S_1(\escript )$ with $J(a_i)=P_{\phi _i}$ and since $J$ is linear we obtain
\begin{equation*}
J\paren{\sum\lambda _ia_i}=\sum\lambda _iJ(a_i)=A
\end{equation*}
Moreover, $a_1\oplus\cdots\oplus a_n=u$ because
\begin{equation*}
J(a_1\oplus\cdots\oplus a_n)=P_{\phi _1}+\cdots +P_{\phi _n}=I=J(u)
\end{equation*}
and $J$ is injective. Hence, $\sum\lambda _ia_i\in\escript$ so $J$ is surjective.
\end{proof}

It follows that if $\escript$ satisfies the conditions of Theorem~\ref{thm33}, then the transition probability has the usual form
$\ahat (b)=\ab{\elbows{\ahat ,\bhat\,}}^2$. We then have the symmetry relation $\ahat (b)=\bhat (a)$ which need not hold for a general
$\escript$.

We have seen in Theorems~\ref{thm31} and~\ref{thm33} that classical convex effect algebras have a single context, while Hilbertian convex effect algebras have an uncountable complete set of contexts. Are there convex effect algebras between these two cases? That is, are there convex effect algebras with only a finite number greater than one, of contexts? We conjecture that the answer is no. Although we have not been able to prove this conjecture in general, we can show it holds for the first few cases. First notice that if
$\escript\ne\sqbrac{0,1}\subseteq\real$ then a context in $\escript$ must have at least two distinct elements. Indeed, if $\brac{a}$ is a context, then $a=1$. If $b\in\escript$, then $b\le1$ so $b=\lambda 1$ for some $\lambda\in\sqbrac{0,1}\subseteq\real$. Hence,
$\escript =\sqbrac{0,1}\subseteq\real$ which is a contradiction.

\begin{thm}    % Theorem 3.4
\label{thm34}
A spectral convex effect algebra $\escript$ does not have exactly two or three mutually disjoint contexts.
\end{thm}
\begin{proof}
Suppose that $\escript$ has exactly two disjoint contexts $\ascript =\brac{a_1,\ldots ,a_n}$, $\bscript =\brac{b_1,\ldots ,b_m}$ with
$n,m\ge 2$. Then
\begin{equation*}
c=\tfrac{1}{2}a_1+\tfrac{1}{2}b_1\le\tfrac{1}{2}1+\tfrac{1}{2}1=1
\end{equation*}
so $c\in\escript$. Since $\escript$ is spectral we can assume without loss of generality that $c=\sum\lambda _ia_i$,
$\lambda _i\in\sqbrac{0,1}$. Now
\begin{equation*}
\ahat _1(c)=\tfrac{1}{2}+\tfrac{1}{2}\ahat _1(b_1)=\lambda _1
\end{equation*}
so we have that $\lambda _1\le 1/2$. Hence,
\begin{equation}                % equation (3.5)
\label{eq35}
\tfrac{1}{2}b_1=\paren{\lambda _1-\tfrac{1}{2}}a_1\oplus\lambda _2a_2\oplus\cdots\oplus\lambda _na_n
\end{equation}
where at least one of the coefficients $\lambda _1-\tfrac{1}{2},\lambda _2,\ldots ,\lambda _n$ is nonzero. If $\lambda _j\ne 0$,
$j\in\brac{2,\ldots ,n}$, then $\lambda _ja_j\le\tfrac{1}{2}b_1$. Since $2\lambda _ja_j\le b_1$ and $b_1\in S_1 (\escript )$ we conclude that $2\lambda _ja_j=\mu b_1$ for some $\mu\in\sqbrac{0,1}$. Let $\alpha =2\lambda _j/\mu$ so $b_1=\alpha a_j$. If $\alpha <1$, letting
$\beta =\min (\alpha ,1-\alpha )$ we obtain
\begin{align*}
\beta a_j&\le\alpha a_j=b_1
\intertext{and}
\beta a_j&\le (1-\alpha )a_j=a_j-\alpha a_j\le 1-\alpha a_j=b'_1
\end{align*}
Since $\beta a_j\ne 0$, this contradicts the fact that $b_1\in S(\escript )$. If $\alpha >1$ we get a similar contradiction. Hence, $\alpha =1$ and $b_1=a_j$ which contradicts the fact that $\ascript\cap\bscript =\emptyset$. If $\lambda _1\ne 1/2$ , we obtain a similar contradiction. We conclude that $\escript$ does not contain two disjoint contexts.

Next suppose that $\escript$ has exactly three mutually disjoint contexts $\ascript =\brac{a_1,\ldots ,a_n}$,
$\bscript =\brac{b_1,\ldots ,b_m}$, $\cscript =\brac{c_1,\ldots c_p}$ with $n,m,p\ge 2$. Then
$d=\tfrac{1}{3}a_1+\tfrac{1}{3}b_1+\tfrac{1}{3}c_1\in\escript$ and as before we can assume that $d=\sum\lambda _ia_1$,
$\lambda _i\in\sqbrac{0,1}$. Since
\begin{equation*}
\ahat _1(d)=\tfrac{1}{3}+\tfrac{1}{3}\ahat _1(b_1)+\tfrac{1}{3}\ahat _1(c_1)=\lambda _1
\end{equation*}
we have that $\lambda _1\ge\tfrac{1}{3}$. Hence,
\begin{equation}                % equation (3.6)
\label{eq36}
\tfrac{1}{3}b_1+\tfrac{1}{3}c_1=\paren{\lambda _1-\tfrac{1}{3}}a_1\oplus\lambda _2a_2\oplus\cdots\oplus\lambda _na_n
\end{equation}
Now $e=\tfrac{1}{2}b_1+\tfrac{1}{3}c_1\in\escript$ but $e$ cannot be spectral relative to $\bscript$ or $\cscript$ because we would obtain an equation like \eqref{eq35} which we saw in the previous paragraph leads to a contradiction. Hence,
\begin{equation}                % equation (3.7)
\label{eq37}
\tfrac{1}{2}b_1+\tfrac{1}{3}c_1=\mu _1a_1\oplus\cdots\oplus\mu _na_n
\end{equation}
with $\mu _i\in\sqbrac{0,1}$. Since
\begin{equation*}
\mu _1a_1+\cdots +\mu _na_n=\tfrac{1}{2}b_1+\tfrac{1}{3}c_1\ge\tfrac{1}{3}b_1+\tfrac{1}{3}c_1
=\paren{\lambda _1-\tfrac{1}{3}}a_1+\lambda _2a_2+\cdots +\lambda _na_n
\end{equation*}
we have that
\begin{equation*}
\mu _1=\ahat _1(\mu _1a_1)\ge\ahat _1\sqbrac{\paren{\lambda _1-\tfrac{1}{3}}a_1}=\lambda _1-\tfrac{1}{3}
\end{equation*}
and similarly $\mu _j\ge\lambda j$, $j=2,\ldots ,n$. Subtracting \eqref{eq36} from \eqref{eq37} gives
\begin{align}                % equation (3.8)
\label{eq38}
\tfrac{1}{6}b_1&=\paren{\tfrac{1}{2}b_1+\tfrac{1}{3}c_1}-\paren{\tfrac{1}{3}b_1+\tfrac{1}{3}c_1}\notag\\
&=\sqbrac{\mu _1-\paren{\lambda _1-\tfrac{1}{3}}}a_1+(\mu _2-\lambda _2)a_2+\cdots +(\mu _n-\lambda _n)a_n
\end{align}
As with \eqref{eq35} in the previous paragraph, we obtain a contradiction. We conclude that $\escript$ does not have three mutually disjoint contexts.
\end{proof}

\section{Convex Sequential Effect Algebras} % Section 4
A convex effect algebra describes the parallel sum $a\oplus b$ and the attenuated scalar product $\lambda a$ for effects. However, there is an important missing ingredient which is the sequential product $a\circ b$. The product $a\circ b$ describes an experiment in which $a$ is measured first and $b$ is measured second. We might say that $a\circ b$ is a measurement of the effect $b$ conditioned by a previous measurement of the effect $a$. Such a temporal or sequential order does not seem to be considered in classical probability theory. For example, if $A$ and $B$ are events in a classical probability space then their intersection $A\cap B$ represents the event that $A$ and $B$ both occur and no consideration is taken for which occurs first. A little more subtle is the conditional probability of $B$ given $A$ described by $P(B|A)=P(A\cap B)/P(A)$. It may appear that $A$ occurs first but we have that
\begin{equation*}
P(A)P(B|A)=P(B)P(A|B)
\end{equation*}
and if it happens that $P(A)=P(B)$ then $P(B|A)=P(A|B)$.

In quantum mechanics $a\circ b$ is useful for describing quantum interference. Because of the sequential order for $a\circ b$, since $a$ is measured first, $a$ may interfere with the $b$ measurement and since $b$ is measured second, $b$ will never interfere with the $a$ measurement. If $a\circ b=b\circ a$ we write $a|b$ and say that $a$ and $b$ \textit{do not interfere}. We now present our general definition.

A \textit{convex sequential effect algebra} (convex SEA) is an algebraic system $(\escript ,0,1,\oplus,\circ )$ where
$(\escript ,0,1,\oplus )$ is an effect algebra and $\circ\colon\escript\times\escript\to\escript$ is a binary operation satisfying:
\begin{list} {(S\arabic{cond})}{\usecounter{cond}%
\setlength\itemindent{-7pt}}
%(S1)
\item $b\mapsto a\circ b$ is additive for all $a\in\escript$.
%(S2)
\item $1\circ a=a$ for all $a\in\escript$.
%(S3)
\item If $a\circ b=0$, then $a|b$.
%(S4)
\item If $a|b$, then $a|b'$ and $a\circ (b\circ c)=(a\circ b)\circ c$ for all $c\in\escript$.
%(S5)
\item If $c|a$ and $c|b$, then $c|a\circ b$ and $c|(a\oplus  b)$ whenever $a\perp b$.
%(S6)
\item For all $\lambda\in\sqbrac{0,1}\subseteq\real$, $a,b\in\escript$, we have that
$(\lambda a)\circ b=a\circ (\lambda b)=\lambda (a\circ b)$.
\end{list}

The next theorem which is proved in \cite{gg02} shows that the sequential product has desirable properties.

\begin{thm}    % Theorem 4.1
\label{thm41}
{\rm (i)}\enspace $a\circ b\le a$ for all $a,b\in\escript$.
{\rm (ii)}\enspace If $a\le b$, then $c\circ a\le c\circ b$ for all $c\in\escript$.
{\rm (iii)}\enspace $a\in S(\escript )$ if and only if $a\circ a=a$.
{\rm (iv)}\enspace For $a\in\escript$, $b\in S(\escript )$, $a\circ b=0$ if and only if $a\perp b$.
{\rm (vi)}\enspace For $a\in\escript$, $b\in S(\escript )$, $a\le b$ if and only if $a\circ b=b\circ a=a$ and $b\le a$ if and only if
$a\circ b=b\circ a=b$.
\end{thm}

A classical effect algebra $\escript (\Omega ,\ascript )$ is a convex SEA under the usual function product $f\circ g=fg$. It is shown in \cite{gg02} that a Hilbertian effect algebra $\escript (H)$ is a convex SEA under the product
\begin{equation*}
A\circ B=A^{1/2}BA^{1/2}
\end{equation*}
where $A^{1/2}$ is the unique positive square root of $A$. It is shown in \cite{gg02} that $A|B$ if and only if $AB=BA$. Of course $\escript (\Omega ,\ascript )$ is commutative while $\escript (H)$ is not where \textit{commutative} means $a\circ b=b\circ a$ for all $a,b$.

A convex SEA has stronger properties than a convex effect algebra. We begin to illustrate this in the following lemma.

\begin{lem}    % Lemma 4.2
\label{lem42}
Let $\escript$ be a convex SEA.
{\rm (i)}\enspace For $a,b\in S_1(\escript )$ we have $a|b$ if and only if $a=b$ or $a\circ b=0$.
{\rm (ii)}\enspace For two contexts $\ascript =\brac{a_1,\ldots ,a_n}$, $\bscript =\brac{b_1,\ldots ,b_m}$ in $\escript$ we have $a_i|b_j$, $i=1,\ldots ,n$, $j=1,\ldots ,m$, if and only if $\ascript =\bscript$.
\end{lem}
\begin{proof}
(i)\enspace If $a=b$ or $a\circ b=0$, then $a|b$ by Theorem~\ref{thm41}(iv). Conversely, suppose that $a|b$. By Theorem~\ref{thm41}(i) we have that $a\circ b\le a,b$ and hence $a\circ b=\lambda a$ and $a\circ b=\mu b$ for some
$\lambda ,\mu\in\sqbrac{0,1}\subseteq\real$. If $\lambda =0$, then $a\circ b=0$. Otherwise, we have that $a=\tfrac{\mu}{\lambda}b$ and squaring gives
\begin{equation*}
a=\paren{\frac{\mu}{\lambda}}^2b=\frac{\mu}{\lambda}b
\end{equation*}
Since $\tfrac{\mu}{\lambda}\ne 0$ we conclude that $\mu =\lambda$. Hence, $a=b$.\newline
(ii)\enspace Since $a_i\perp a_j$ for $i\ne j$, by Theorem~\ref{thm41}(iv) $a_i\circ a_j=0$ for $i\ne j$. Hence, $a_i|a_j$, $i,j=1,\ldots ,n$, by (S3). We conclude that if $\ascript =\bscript$ then $a_i|b_j$. Conversely, suppose $a_i|b_j$, $i=1,\ldots ,n$, $j=1,\ldots ,m$. Since $b_1\oplus\cdots\oplus b_m=1$, by (S1) we have
\begin{equation*}
a_i=a_i\circ b_1\oplus\cdots\oplus a_i\circ b_m
\end{equation*}
If $a_i\circ b_j=0$ for $j=1,\ldots ,m$, then $a_i=0$ which is a contradiction. Hence, $a_i\circ b_j\ne 0$ for some $j=1,\ldots ,m$. By (i) of this lemma, $a_i=b_j$. It follows that $m=n$ and $\ascript =\bscript$.
\end{proof}

In the sequel, we shall assume that $\escript$ is a finite dimensional convex SEA. If $\escript$ is commutative, then $\escript$ is classical. Indeed, it follows from Lemma~\ref{lem42}(ii) that $\escript$ possesses exactly one context $\ascript =\brac{a_1,\ldots ,a_n}$. Moreover, if $b\in\escript$ then by Theorem~\ref{thm41}(i) we have
\begin{equation*}
b=b\circ a_1\oplus\cdots\oplus b\circ a_n=a_1\circ b\oplus\cdots\oplus a_n\circ b=\lambda _1a_1\oplus\cdots\oplus\lambda _na_n
\end{equation*}
for $\lambda _i\in\sqbrac{0,1}$, $i=1,\ldots ,n$. It follows that $\escript$ is spectral so by Theorem~\ref{thm31}, $\escript$ is classical as an effect algebra. To show that $\escript$ is classical as a SEA, consider the isomorphism $J\colon\escript\to\escript (\Omega ,\ascript )$ of Theorem~\ref{thm31}. If $b\in\escript$ is given as before we have $J(b)(\omega _i)=\lambda _i$, $i=1,\ldots ,n$. If $c\in\escript$ with $c=\mu_1a_1\oplus\cdots\oplus\mu _na_n$, then 
\begin{equation*}
J(b\circ c)(\omega _i)=\lambda _i\mu _i=J(b)(\omega _i)J(c)(\omega _i)=J(b)J(c)(\omega _i)
\end{equation*}
Hence, $J$ is a SEA isomorphism so $\escript$ is a classical SEA.

For $\ascript =\brac{a_1,\ldots ,a_n}$ with $a_i\in S(\escript )$ and $\sum a_i=1$, we say that $a\in\escript$ is
$\ascript$-\textit{measurable} if
\begin{equation}                % equation (4.1)
\label{eq41}
a=\sum _{i=1}^n\lambda _ia_i
\end{equation}
It follows from Theorem~\ref{thm41}(iv) that $a_i\circ a_j=0$ for $i\ne j$. It also follows from Theorem~\ref{thm41}(iv) that if
$a,b\in S(\escript )$ with $a\perp b$, then $a\oplus b\in S(\escript )$. Hence, we can and will assume without loss of generality that
$\lambda _i\ne\lambda _j$, $i\ne j$, in \eqref{eq41}. For $a\in\escript$, we define $a^0=1$ and 
\begin{equation*}
a^i=a\circ a\circ\cdots\circ a\quad(i\hbox{ factors})
\end{equation*}
An effect $b\in\escript$ is a \textit{function of} $a\in\escript$ if
\begin{equation*}
b=\sum _{i=1}^n\alpha _ia^i,\quad\alpha _i\in\real
\end{equation*}
Notice that some of the $\alpha _i$ can be negative and we can even have $\alpha _i>1$ or $\alpha _i<-1$, but the sum is still in
$\escript$. The individual terms in the sum can be thought of being in the encompassing ordered vector space $(V,K)$. For example, $b'=1-b$ and 
\begin{equation*}
(b')^2=(1-b)\circ(1-b)=1-2b+b^2
\end{equation*}
so $b'$ and $(b')^2$ are functions of $b$. For another example, if $a\perp a$ then $b=a\oplus a=2a$ so $b$ is a function of $a$ and
\begin{equation*}
b'=1-b=1-2a
\end{equation*}
is again a function of $a$. Notice that if $a|b$, then any function of $a$ commutes with any function of $b$.

If $b_1$ and $b_2$ are functions of $a$, then $b_1\circ b_2=b_2\circ b_1$ is a function of $a$ and $b_1\oplus b_2$ is a function of $a$ whenever $b_1\perp b_2$. Also, $0,1$ and $\lambda a$, $\lambda\in\sqbrac{0,1}\subseteq\real$ are functions of $a$. It follows that the functions of $a$ form a commutative sub-convex SEA of $\escript$. Suppose $a=\lambda _1a_1+\lambda _2a_2$ is
$\brac{a_1,a_2}$-measurable so that $a_1,a_2\in S(\escript )$, $a_1+a_2=1$, and $\lambda _1\ne\lambda _2$. We now show that $a_1$ and $a_2$ are functions of $a$. Since $a_1=1-a_2$ we have that
\begin{equation*}
a=\lambda _1(1-a_2)+\lambda _2a_2=\lambda _11+(\lambda _2-\lambda _1)a_2
\end{equation*}
Hence,
\begin{align*}
a_2=\frac{a-\lambda _11}{\lambda _2-\lambda _1}\\
\intertext{and}
a_1=1-a_2=\frac{\lambda _21-a}{\lambda _2-\lambda _1}
\end{align*}
so $a_1$ and $a_2$ are functions of $a$. We now generalize this result.

\begin{thm}    % Theorem 4.3
\label{thm43}
{\rm (i)}\enspace If $a=\sum\lambda _ia_i$ is $\brac{a_1,\ldots ,a_n}$-measurable, then $a_i$ is a function of $a$, $i=1,\ldots ,n$.
{\rm (ii)}\enspace  Also, if $b$ is $\brac{b_i}$-measurable and $a|b$ then $a\circ b$ is $\brac{a_i\circ b_j}$-measurable and $a\oplus b$ is $\brac{a_i\circ b_j}$-measurable whenever $a\perp b$.
\end{thm}
\begin{proof}
(i)\enspace If $a=\sum _{i=1}^n\lambda _ia_i$ we obtain the system of equations
\begin{align*}
&a_1+a_2+\cdots +a_n=1\\
&\lambda _1a_1+\lambda _2a_2+\cdots +\lambda _na_n=a\\
&\lambda _1^2a_1+\lambda _2^2a_2+\cdots +\lambda _n^2a_n=a^2\\
&\quad\vdots\\
&\lambda _1^{n-1}a_1+\lambda _2^{n-1}a_2+\cdots +\lambda _n^{n-1}a_n=a^{n-1}\\
\end{align*}
the determinant for this system is the Vandermonde determinant
\begin{equation*}
\left|\begin{matrix}\noalign{\smallskip}1&1&\cdots&1\\\noalign{\smallskip}\lambda  _1&\lambda _2&\cdots&\lambda _n\\
\noalign{\smallskip}\lambda  _1^2&\lambda _2^2&\cdots&\lambda _n^2\\
\noalign{\smallskip}\vdots&&&\\
\noalign{\smallskip}\lambda  _1^{n-1}&\lambda _2^{n-1}&\cdots&\lambda _n^{n-1}\end{matrix}\right|
=(-1)^{n(n-1)/2}\prod _{i<j}(\lambda _i-\lambda _j)
\end{equation*}
Since $\lambda _i\ne\lambda _j$, $i\ne j$, the determinant is nonzero. Hence, there is a unique solution to this system of equations for the unknowns $a_i$, $i=1,\ldots ,n$. We conclude that $a_i$ is a function of $a$, $i=1,\ldots ,n$.\newline
(ii)\enspace Suppose $a$ and $b$ are $\brac{a_i}$ and $\brac{b_i}$-measurable and $a|b$. Then by (i) of this theorem we have
$a=\sum\lambda _ia_i$, $b=\sum\mu _ib_i$ where the $a_i$ are functions of $a$ and the $b_i$ are functions of $b$. Since $a|b$, any function of $a$ commutes with any function of $b$. Hence, $a_i|b_j$ for all $i,j$. But then
\begin{equation*}
a\circ b=\sum\lambda _i\mu _ia_i\circ b_j
\end{equation*}
where $\sum a_i\circ b_j=1$ and $a_i\circ b_j\in S(\escript )$ by Theorem~\ref{thm41}(iii). Hence, $a\circ b$ is
$\brac{a_i\circ b_j}$-measurable. If we also have $a\perp b$, then
\begin{equation*}
a\oplus b=\sum\lambda _ia_i+\sum\mu _jb_j=\sum _{i,j}\lambda _ia_i\circ b_j+\sum _{i,j}\mu _ja_i\circ b_j
\end{equation*}
Hence, $a\oplus b$ is $\brac{a_i\circ b_j}$-measurable.
\end{proof}

We now apply Theorem~\ref{thm43} to obtain a strengthening of Theorem~\ref{thm34} for a convex SEA.

\begin{cor}    % Corollary 4.4
\label{cor44}
A spectral convex SEA $\escript$ does not have exactly two, three or four mutually disjoint contexts.
\end{cor}
\begin{proof}
Theorem~\ref{thm34} treats the two and three mutually disjoint contexts cases. Now suppose $\escript$ possesses exactly four mutually disjoint contexts $\ascript =\brac{a_i}$, $\bscript =\brac{b_i}$, $\cscript =\brac{c_i}$ and $\dscript =\brac{d_i}$. As in
Theorem~\ref{thm34} we have that
\begin{equation*}
e=\tfrac{1}{4}\,a_1+\tfrac{1}{4}\,b_1+\tfrac{1}{4}\,c_1+\tfrac{1}{4}\,d_1=\escript
\end{equation*}
and we can assume without loss of generality that $e=\sum\lambda _ia_i$, $\lambda _i\in\sqbrac{0,1}$ which gives
\begin{equation}                % equation (4.2)
\label{eq42}
\tfrac{1}{4}b_1+\tfrac{1}{4}c_1+\tfrac{1}{4}d_1=\paren{\lambda _1-\tfrac{1}{4}}a_1+\lambda _2a_2+\cdots +\lambda _na_n
\end{equation}
Now $\tfrac{1}{4}c_1+\tfrac{1}{4}d_1$ cannot be spectral relative to $\cscript$, $\dscript$ or $\ascript$ because we would obtain a contradiction as with \eqref{eq35} in Theorem~\ref{thm34}. We therefore have that
\begin{equation*}
\tfrac{1}{4}\,c_1+\tfrac{1}{4}\,d_1=\sum\mu _ib_i
\end{equation*}
so by \eqref{eq42} we obtain
\begin{equation}                % equation (4.3)
\label{eq43}
b=\paren{\mu _1+\tfrac{1}{4}}b_1+\mu 2b_2+\cdots +\mu _mb_m
=\paren{\lambda _1-\tfrac{1}{4}}a_1+\lambda _2a_2+\cdots +\lambda _na_n
\end{equation}
By considering the coefficients in \eqref{eq43} that are different we can apply Theorem~\ref{thm43} to conclude that $b_i|a_j$ and
$b_i\circ a_j\ne 0$ for some $i$ and $j$. It follows from Lemma~\ref{lem42} that $b_i=a_j$. This contradicts the fact that
$\ascript\cap\bscript =\emptyset$.
\end{proof}

If $b=\sum\lambda _ia_i$ for a context $\ascript =\brac{a_i\colon i=1,\ldots ,n}$, then $b$ is $\ascript$-measurable and the results of Theorem~\ref{thm43} hold. Moreover, $b$ is $\ascript$-measurable if and only if $b|a_i$, $i=1,\ldots ,n$. Indeed, if
$b=\sum\lambda _ia_i$ then clearly, $b|a_i$, $i=1,\ldots ,n$. Conversely, if $b|a_i$, $i=1,\ldots ,n$, then 
\begin{equation*}
b=\sum b\circ a_i=\sum a_i\circ b=\sum\lambda _ia_i
\end{equation*}
We now discuss Theorems~\ref{thm31} and \ref{thm33} in the case of a convex SEA.

Let $\escript ,\fscript$ be convex SEA's with sequential products $a\circ b$ and $a\ctimes b$, respectively. A SEA \textit{isomorphism} for $\escript$ to $\fscript$ is a convex effect algebra isomorphism $L\colon\escript\to\fscript$ that satisfies $L(a\circ b)=(La)\ctimes (Lb)$ for all $a,b\in\escript$. As we have seen, the map $J$ in Theorem~\ref{thm31} is a SEA isomorphism so that theorem characterizes convex SEA's that are isomorphic to a finite classical SEA. The situation for Hilbertian SEA's is more complicated. Let $\escript$ be a SEA satisfying the conditions of Theorem~\ref{thm33} and let $J\colon\escript\to\escript\paren{\hscript (\bscript )}$ be the convex effect algebra isomorphism of that theorem. Recall that for the chosen context $\bscript$ if $b=\sum\lambda _ia_i$ where
$\ascript =\brac{a_i\colon i=1,\ldots ,n}$ is some context, then
\begin{equation*}
J(b)=\sum\lambda _iP(U_{\ascript\bscript}\ahat _i)
\end{equation*}
The next lemma shows that if $a|b$ then $J(a\circ b)=J(a)J(b)$ where $J(a)J(b)$ is the usual operator product.

\begin{lem}    % Lemma 4.5
\label{lem45}
We have $a|b$ if and only if $J(a)J(b)=J(b)J(a)$. Moreover, if $a|b$ then $J(a\circ b)=J(a)J(b)$.
\end{lem}
\begin{proof}
Suppose that $a|b$ where $a=\sum\lambda _ia_i$, $b=\sum\mu _ib_i$ for contexts $\ascript =\brac{a_i}$, $\cscript =\brac{b_j}$. It follows from Theorem~\ref{thm43}(i) that $a_i|b$, for all $i,j$. Applying Lemma~\ref{lem42}(ii) we conclude that $\ascript =\cscript$. By changing the order of the $\mu_i$'s, we can assume that $b=\sum\mu _ia_i$. As in Theorem~\ref{thm43}(ii) we have that
$a\circ b=\sum\lambda _i\mu _ia_i$. Therefore,
\begin{align*}
J(a\circ b)&=\sum\lambda _i\mu _iP\paren{U_{\ascript\bscript}(\ahat _i)}
  =\sum\lambda _iP\paren{U_{\ascript\bscript}(\ahat _i)}\sum\mu _iP\paren{U_{\ascript\bscript}(\ahat _i)}\\
  &=J(a)J(b)=J(b)J(a)
\end{align*}
Conversely, suppose that $J(a)J(b)=J(b)J(a)$. Since
\begin{align*}
J(a)=\sum\lambda _iP\paren{U_{\ascript\bscript}(\ahat _i)},J(b)=\sum\mu _iP\paren{U_{\cscript\bscript}(\,\bhat _i)}
\end{align*}
as before, we have that
\begin{equation*}
\brac{P\paren{U_{\ascript\bscript}(\ahat _i)}}=\brac{P\paren{U_{\cscript\bscript}(\,\bhat _i)}}
\end{equation*}
Since $J$ is injective we conclude that $\ascript =\cscript$. Hence, $a|b$.
\end{proof}

It follows from Lemma~\ref{lem45} that $a$ is sharp if and only if $J(a)$ is sharp. We now define a product on
$\escript (\hscript (\bscript ))$ induced by the sequential product on $\escript$. If $A,B\in\escript (\hscript (\bscript ))$ are given by
$A=J(a)$, $B=J(b)$ we define $A\ctimes B=J(a\circ b)$. We then have
\begin{equation*}
J(a\circ b)=J(a)\ctimes J(b)
\end{equation*}
by definition. The next result shows that $A\ctimes B$ is a sequential product.

\begin{thm}    % Theorem 4.6
\label{thm46}
With the product $A\ctimes B$, $\escript (\hscript (\bscript ))$ is a convex SEA
\end{thm}
\begin{proof}
We assume that $J(a)=A$, $J(b)=B$, $J(c)=C$, $J(b_1)=B_1$ and $J(b_2)=B_2$. We now check the six axioms for a convex SEA.
\begin{list} {(S\arabic{cond})}{\usecounter{cond}%
\setlength\itemindent{-7pt}}
%(S1)
\item Since $J(b_1\oplus b_2)=J(b_1)\oplus J(b_2)=B_1\oplus B_2$ we have
\begin{align*}
\hskip -2pc A\ctimes (B_1\oplus B_2)&=J(a\circ (b_1\oplus b_2))=J(a\circ b_1\oplus a\circ b_2)=J(a\circ b_1)+J(a\circ b_2)\\
  &=A\ctimes B_1\oplus A\ctimes B_2
\end{align*}
%(S2)
\item $I\ctimes A=J(1\circ a)=J(a)=A$
%(S3)
\item If $A\ctimes B=0$, the $J(a\circ b)=0$. Since $J$ is injective, $a=b=0$ so $a|b$. Hence, $A\ctimes B=B\ctimes A$ by
Lemma~\ref{lem45}.
%(S4)
\item If $A\ctimes B=B\ctimes A$, then $A\ctimes B'=B'\ctimes A$. Moreover, since $a|b$ we have
\begin{align*}
\hskip -2pc A\ctimes (B\ctimes C)&=A\ctimes J(b\circ c)=J\sqbrac{a\circ (b\circ c)}=J\sqbrac{(a\circ b)\circ c}=J(a\circ b)\ctimes J(c)\\
  &=\sqbrac{J(a)\ctimes J(b)}\ctimes J(c)=(A\ctimes B)\ctimes C
\end{align*}
%(S5)
\item If $C\ctimes A=A\ctimes C$ and $C\ctimes B=B\ctimes C$ then by Lemma~\ref{lem45}, $c|a$ and $c|b$ so we have that $c|(a\circ b)$ and $c|(a\oplus b)$. Therefore, $J(c)|J(a\circ b)$ so $C|A\ctimes B$ and $J(c)|J(a\oplus b)$ so $C|(A+B)$.
%(S6)
\item If $\lambda\in\sqbrac{0,1}\subseteq\real$, then
\begin{equation*}
(\lambda A)\ctimes B=J(\lambda a\circ b)=\lambda J(a\circ b)=\lambda (A\ctimes B)
\end{equation*}
and similarly, $A\ctimes (\lambda B)=\lambda (A\ctimes B)$.\qedhere
\end{list}
\end{proof}

It follows from Theorem~\ref{thm46} that $J$ is a SEA isomorphism from $\escript$ to $\escript\paren{\hscript (\bscript )}$. We have not proved that $A\ctimes B$ is the standard sequential product $A\circ B=A^{1/2}BA^{1/2}$. A characterization of when $A\ctimes B=A\circ B$ are the following physically justifiable conditions \cite{gl08}:

\begin{list} {(B\arabic{cond})}{\usecounter{cond}
\setlength{\rightmargin}{\leftmargin}}
%(B1)
\item For every density operator $\rho$ and $A,B\in\escript\paren{\hscript (\bscript )}$ we have
\begin{equation*}
\rmtr\sqbrac{(A\ctimes\rho )B}=\rmtr\sqbrac{\rho (A\ctimes B)}
\end{equation*}
%(B2)
\item If $P$ is a one-dimensional projection in $\escript\paren{\hscript (B)}$ and $A\in\escript (\hscript (\bscript ))$ with $A\ctimes P\ne 0$ then $A\ctimes P/\rmtr (A\ctimes P)$ is a one-dimensional projection.
\end{list}

It has been very important in our previous work that if $\escript$ is a convex effect algebra and $a\in S(\escript )$ then there exists a state
$\ahat\in\Omega (\escript )$ such that $\ahat (a)=1$. We now show that if $\escript$ is a convex SEA, then we can construct this state explicitly. For $b\in\escript$, since $a\circ b\le a$, there exists a $\lambda (a,b)\in\sqbrac{0,1}\subseteq\real$ such that
$a\circ b=\lambda (a,b)a$. Since $\lambda (a,1)=1$ and
\begin{align*}
\lambda (a,b_1\oplus b_2)a&=a\circ (b_1\oplus b_2)=a\circ b_1\oplus a\circ b_2=\lambda (a,b_1)a\oplus\lambda (a,b_2)a\\
   &=\sqbrac{\lambda (a,b_1)+\lambda (a,b_2)}a
\end{align*}
we conclude that $\lambda (a,b_1\oplus b_2)=\lambda (a,b_1)+\lambda (a,b_2)$. Hence, $b\mapsto\lambda (a,b)$ is a state satisfying
$\lambda (a,a)=1$. We then use the notation
\begin{equation*}
\ahat (b)=\lambda (a,b)
\end{equation*}
for all $b\in\escript$. Notice that $\ahat (a\circ b)=\ahat (b)$ for all $b\in\escript$.

In the sequel, $\escript$ will denote a convex SEA with order determining set of states $\Omega (\escript )$. One of the advantages of working with a SEA is that it provides a structure for defining conditional probabilities. If $\omega\in\Omega (\escript )$ and $a\in\escript$ with
$\omega (a)\ne 0$, then the state $\omega$ \textit{conditioned by} $a$ is
\begin{equation*}
\omega (b|a)=\omega (a\circ b)/\omega (a)
\end{equation*}
for all $b\in\escript$. Notice that $\omega (a\circ b)=\omega (a)\omega (b|a)$. When we write $\omega (b|a)$ we are implicitly assuming that
$\omega (a)\ne 0$.

\begin{lem}    % Lemma 4.7
\label{lem47}
{\rm (i)}\enspace For every $\omega\in\Omega (\escript )$ and $a\in S_1(\escript )$ we have that $\omega (b|a)=\ahat (b)$ for all
$b\in\escript$.
{\rm (ii)}\enspace $a\in S(\escript )$ if and only if $\omega (a|a)=1$ for all $\omega\in\Omega (\escript )$.
\end{lem}
\begin{proof}
(i)\enspace For $a\in S_1(\escript )$ we have that
\begin{equation*}
\omega (b|a)=\frac{\omega (a\circ b)}{\omega (a)}=\frac{\omega (\,\ahat\, (b)a)}{\omega (a)}=\ahat (b)
\end{equation*}
(ii)\enspace If $a\in S(\escript )$ then
\begin{equation*}
\omega (a|a)=\frac{\omega (a\circ a)}{\omega (a)}=1
\end{equation*}
for every $\omega\in\Omega (\escript )$. Conversely, if $\omega (a|a)=1$ for all $\omega$ with $\omega (a)\ne 0$, then
\begin{equation*}
\omega (a^2)=\omega (a)\omega (a|a)=\omega (a)
\end{equation*}
Clearly, $\omega (a^2)=\omega (a)$ if $\omega (a)=0$. Since $\Omega (\escript )$ is separating $a^2=a$ so $a\in S_1(\escript )$.
\end{proof}

Lemma~\ref{lem47}(i) shows that all states conditioned by an $a\in S_1(\escript )$ are the same. In this sense, $\ahat$ is universal.

A \textit{measurement} is a set $\ascript =\brac{a_1,\ldots ,a_n}\subseteq\escript$ satisfying $a_1\oplus\cdots\oplus a_n=1$. We say that
$b\in\escript$ is \textit{measurable relative to} $\ascript$ if $b$ has the form $b=\sum\lambda _ia_i$, $\lambda _i\in\sqbrac{0,1}$. We say that $\ascript$ is a \textit{sharp measurement} if $a_i\in S(\escript )$, $i=1,\ldots ,n$. We have already treated sharp measurements and in this case measurable relative to $\ascript$ and $\ascript$-measurable are the same. The \textit{law of total probability} for $\omega\in\Omega (\escript )$ says if $b\in\escript$ and $\ascript =\brac{a_i\colon i=1,\ldots ,n}$ is a measurement, then
\begin{equation*}
\omega (b)=\sum\omega (a_i\circ b)=\sum\omega (a_i)\omega (b|a_i)
\end{equation*}
This law holds for some $\omega$, $b$ and $\ascript$ and not for others as the following lemma shows.

\begin{lem}    % Lemma 4.8
\label{lem48}
{\rm (i)}\enspace If $b|a_i$, $i=1,\ldots ,n$, then $\omega (b)=\sum\omega (a_i\circ b)$ for every $\omega\in\Omega (\escript )$.
{\rm (ii)}\enspace If $\ascript$ is sharp and $\omega (b)=\sum\omega (a_i\circ b)$ for every $\omega\in\escript$, then $b|a_i$, $i=1,\ldots ,n$.
\end{lem}
\begin{proof}
(i)\enspace If $b|a_i$, $i=1,\ldots ,n$, since $b=\sum b\circ a_i$ we have that
\begin{equation*}
\omega (b)=\omega \paren{\sum b\circ a_i}=\sum\omega (b\circ a_i)=\sum\omega (a_i\circ b)
\end{equation*}
(ii)\enspace Assume that $\ascript$ is sharp and $\omega (b)=\sum\omega (a_i\circ b)$ for all $\omega\in\Omega (\escript )$. We then have that $\omega (b)=\omega\paren{\sum a_i\oplus b}$ for all $\omega\in\Omega (\escript )$. Since $\Omega (\escript )$ is separating, we conclude that $b=\sum a_i\circ b$. Since
\begin{equation*}
a_i\circ b\le a_i\le a'_j,\quad i\ne j
\end{equation*}
it follows from Theorem~\ref{thm41}(v) that $a_i\circ b|a_i$ and $a_i\circ b|a'_j$ for $j\ne i$. Hence, $a_i\circ b|a_j$, $j=1,\ldots ,n$. Therefore,
\begin{equation*}
b\circ a_j=\paren{\sum a_i\circ b}\circ a_j=a_j\circ\paren{\sum a_i\circ b}=\sum _{i=1}^n(a_j\circ a_i)\circ b=a_j\circ b
\end{equation*}
$j=1,\ldots ,n$.
\end{proof}

In a similar vein, \textit{Bayes' Rule} for $\omega\in\Omega (\escript )$ says that if $b\in\escript$ and $\ascript =\brac{a_i\colon i=1,\ldots ,n}$ is a measurement, then
\begin{equation*}
\omega (a_i|b)=\frac{\omega (b|a_i)\omega (a_i)}{\omega (b)}
\end{equation*}
It immediately follows that Bayes' Rule holds for all $\omega\in\Omega (\escript )$ if and only if $b|a_i$, $i=1,\ldots ,n$.

If $\omega\in\Omega (\escript )$ and $\ascript =\brac{a_i\colon i=1,\ldots ,n}$ is a measurement, the \textit{conditional expectation of}
$b\in\escript$ \textit{given} $\ascript$ is an effect denoted by $E_\omega (b|\ascript )$ that is measurable relative to $\ascript$ and satisfies 
\begin{equation*}
\omega\sqbrac{a\circ E_\omega (b|\ascript )}=\omega (a\circ b)
\end{equation*}
for all $a\in\ascript$. Notice that $b$ is measurable relative to $\ascript$ if and only if $E_\omega (b|\ascript )=b$.

\begin{thm}    % Theorem 4.9
\label{thm49}
{\rm (i)}\enspace The map $b\mapsto E_\omega (b|\ascript )$ is affine and additive. 
{\rm (ii)}\enspace If $\ascript$ is sharp, $b\mapsto E_\omega (b|\ascript )$ is a morphism.
{\rm (iii)}\enspace If $a_i\in S_1(\escript )$, $i=1,\ldots ,n$, then $E_{\,\ahat _i}(b|\ascript )=\ahat _i(b|a _i)$ and 
\begin{equation*}
E_\omega (b|\ascript )=\sum\brac{\ahat _i(b)a_i\colon\omega (a_i)\ne 0}
\end{equation*}
\end{thm}
\begin{proof}
(i)\enspace Since $E_\omega (b|\ascript )$ is measurable relative to $\ascript$, clearly $\lambda E_\omega (b|\ascript )$ is also for
$\lambda\in\sqbrac{0,1}$. Moreover, for $a\in\ascript$ we have
\begin{align*}
\omega\sqbrac{a\circ\lambda E_\omega (b|\ascript )}&=\lambda\omega\sqbrac{a\circ E_\omega (b|\ascript )}=\lambda\omega (a\circ b)\\
&=\omega (a\circ\lambda b)=\omega\sqbrac{a\circ E_\omega (\lambda b|\ascript )}
\end{align*}
Hence, $E_\omega (\lambda b|\ascript =\lambda E_\omega (b|\ascript )$. If $b_1\perp b_2$, then clearly
$E_\omega (b_1|\ascript )\perp E_\omega (b_2|\ascript )$. Moreover, for $a\in\ascript$ we have 
\begin{align*}
\omega\sqbrac{a\circ E_\omega (b_1\oplus b_2)(a)}&=\omega\sqbrac{a\circ (b_1\oplus b_2)}=\omega (a\circ b_1)+\omega (a\circ b_2)\\
  &=\omega\sqbrac{a\circ E_\omega (b_1|\ascript)}+\omega\sqbrac{a\circ E_\omega (b_2|\ascript )}\\
  &=\omega\brac{a\circ\sqbrac{E_\omega (b_1|\ascript )\oplus E_\omega (b_2|\ascript )}}
\end{align*}
We conclude that
\begin{equation*}
E_\omega (b_1\oplus b_2|\ascript )=E_\omega (b_1|\ascript )\oplus E_\omega (b_2|\ascript )
\end{equation*}
(ii)\enspace Suppose $\ascript$ is sharp and $E_\omega (b|\ascript )=\sum\lambda _ia_i$. We then have
\begin{equation*}
\omega (a_j\circ b)=\omega\sqbrac{a_j\circ E_\omega (b|\ascript )}=\omega\paren{\sum\lambda _ia_j\circ a_i}=\lambda _j\omega (a_j)
\end{equation*}
Hence, $\lambda _j=\omega (b|a_j)$ and we have
\begin{equation}                % equation (4.4)
\label{eq44}
E_\omega (b|\ascript )=\sum\omega (b|a_i)a_i
\end{equation}
In particular $E_\omega (1|\ascript )=\sum a_i=1$ so $E_\omega (\cdot|\ascript )$ is a morphism.\newline
(iii)\enspace This follows from \eqref{eq44}.
\end{proof}

\begin{thm}    % Theorem 4.10
\label{thm410}
Let $\ascript$ be a sharp measurement.
{\rm (i)}\enspace If $c$ is measurable relative in $\ascript$, then for all $b\in\escript$ we have
\begin{equation*}
E_\omega (c\circ b|\ascript )=c\circ E_\omega (b|\ascript )
\end{equation*}
{\rm (ii)}\enspace $E_\omega (b|\ascript )=\sum E_\omega (a_i\circ b|\ascript )$
\end{thm}
\begin{proof}
(i)\enspace Clearly, $c\circ E_\omega (b|\ascript )$ is measurable relative to $\ascript$. If $c=\sum\lambda _ja_j$, then by \eqref{eq44} we have 
\begin{align*}
E_\omega (c\circ b|\ascript )&=\sum\omega (c\circ b|a_i)a_i=\sum\frac{\omega\sqbrac{(a_i\circ c)\circ b}}{\omega (a_i)}\,a_i\\
&=\sum\frac{\omega (\lambda _ia_i\circ b)}{\omega (a_i)}=\sum\lambda _i\omega (b|a_i)a_i\\
&=c\circ\sqbrac{\sum\omega (b|a_i)a_i}=c\circ E_\omega (b|\ascript )
\end{align*}
(ii)\enspace By (i) of this theorem, we have that
\begin{align*}
E_\omega (b|\ascript )&=E_\omega (b|\ascript )\circ\sum a_i=\sum E_\omega (b|\ascript )\circ a_i\\
&=\sum a_i\circ E_\omega (b|\ascript )=\sum E_\omega (a_i\circ b|\ascript )\qedhere
\end{align*}
\end{proof}

\end{document}